\documentclass[11pt,twocolumn]{article}
\usepackage{amsmath,amsthm}
\usepackage{rotating}
\usepackage{xcolor}
\usepackage{graphicx}

\usepackage{amsfonts}
\newtheorem{theorem}{Theorem}[section]

\newtheorem{proposition}[theorem]{Proposition}

\theoremstyle{remark}

\theoremstyle{definition}

\theoremstyle{example}

\theoremstyle{notation}

\newcommand{\bra}[1]{\langle#1|}
\newcommand{\ket}[1]{|#1\rangle}

\begin{document}

\title{Uncertainty relations in terms of the Gini index for finite quantum systems }            
\author{A. Vourdas\\Department of Computer Science,\\
University of Bradford, \\
Bradford BD7 1DP, United Kingdom\\a.vourdas@bradford.ac.uk}

\maketitle

\section*{Abstract}
Lorenz values and the Gini index are popular quantities in Mathematical Economics, and are used here in the context of quantum systems with finite-dimensional Hilbert space.
They quantify the uncertainty in the probability distribution related to an orthonormal basis.
It is shown that Lorenz values are superadditive functions and the Gini indices are subadditive functions.
The supremum over all density matrices of  the sum of the two Gini indices with respect to position and momentum states, is used to define an uncertainty coefficient
which quantifies the uncertainty in the quantum system. It is shown that the uncertainty coefficient is positive, and an
upper bound for it is given.
Various examples demonstrate these ideas.

{\bf PACS:} 03.65.Aa

\section{Introduction}
Uncertainty relations for quantum systems with finite-dimensional Hilbert space have been studied in \cite{U1,U2,U3,U4,U5,U6}.
Other uncertainty relations in various contexts have been studied in \cite{UN1,UN2,UN3,UN4,UN5}.
In this paper we present an alternative approach which is based on the Lorenz values and the Gini index, which are quantities used extensively in Mathematical Economics 
for the study of wealth distributions\cite{Gini,Gini1}. Here we use these quantities for the study of probability distributions that are the outcomes of quantum measurements related to an orthonormal basis.

In a recent paper \cite{new} we have used the Lorenz values and the Gini index for the study of $Q$-functions. In this paper we use these quantities for the study of probability distributions
related to position and momentum states. 
Several properties of the Lorenz values and the Gini index in a quantum context, are proven.

 A qualitative uncertainty principle, is that the probability distributions related to positions and momenta cannot be both `narrow'. This is quantified with an uncertainty coefficient
that is defined in terms of the supremum over all density matrices, of the sum of the two Gini indices for the position and momentum states.
We prove in proposition \ref{pro134} that the uncertainty coefficient is a positive number, and we give an upper bound for it.

In section 2 we present briefly some basic concepts for finite quantum systems, in order to define the notation.
In section 3 we introduce Lorenz values in a quantum context and discuss their properties.
We also introduce the concept of comonotonicity which has been used in the context of Choquet integrals with applications in Mathematical Economics and Artificial 
Intelligence\cite{D0,D1,D2,D3}.
Here we use comonotonicity in connection with Lorenz values in a quantum context. 

In section 4 we introduce the Gini index in a quantum context and discuss its properties.
We then introduce an uncertainty coefficient $\eta(d)$ that quantifies the uncertainty principle.
We show that $\eta(d)$ is greater than zero, and give an upper limit for it (Eq.(\ref{F2})).
In section 5 we discuss an example.
We conclude in section 6 with a discussion of our results.

\section{Quantum systems with variables in ${\mathbb Z}_d$ and measurements}

We consider a quantum system with variables in ${\mathbb Z}_d$, the ring of integers modulo $d$ where $d$ is an odd integer (e.g.,\cite{V1}).
$H_d$ is the $d$-dimensional Hilbert space describing these systems. 
$|X;r\rangle$ where $r\in {\mathbb Z}_d$, is an orthonormal basis which we call position basis (the $X$ in this notation is not a variable, but it simply indicates position basis).
Through a Fourier transform we get another orthonormal basis that we call momentum basis:
\begin{eqnarray}
&&|{P};r\rangle=F|{X};r\rangle;\;\;
F=\frac{1}{\sqrt d}\sum _{r,s}\omega^{rs}\ket{X;r}\bra{X;s}\nonumber\\
&&\omega=\exp \left (i\frac{2\pi }{d}\right );\;\;r,s\in {\mathbb Z}_d.
\end{eqnarray}
The displacement operators $Z^\alpha, X^\beta$ in the phase space ${\mathbb Z}_d\times {\mathbb Z}_d$, are given by
\begin{eqnarray}\label{2A}
&&Z^\alpha= \sum \omega^{\alpha m}|X; m\rangle\bra{X;m}\nonumber\\
&&X^\beta =\sum |X; m+\beta\rangle \bra{X; m}.
\end{eqnarray}
where $\alpha$, $\beta \in {\mathbb Z}_d$. 
General displacement operators are the unitary operators
\begin{eqnarray}
&&D(\alpha, \beta)=Z^\alpha X^\beta \omega^{-2^{-1}\alpha \beta}\nonumber\\&&[D(\alpha, \beta)]^{\dagger}=D(-\alpha, -\beta);\;\;\;\alpha, \beta \in {\mathbb Z}_d
\end{eqnarray}
The $2^{-1}=\frac{d+1}{2}$ exists in ${\mathbb Z}_d$ because $d$ is an odd integer.
Coherent states in this context are defined as
\begin{equation}\label{bba}
\ket{\alpha, \beta}_{\rm coh}=D(\alpha, \beta)\ket{f}
\end{equation}
where $\ket{f}$ is a fiducial vector, different from position or momentum states so that we get a non-trivial set of $d^2$ states.
They obey the resolution of the identity
\begin{equation}
\frac{1}{d}\sum_{\alpha,\beta} \ket{\alpha, \beta}_{\rm coh}  \;_{\rm coh}\bra{\alpha, \beta}={\bf 1}.
\end{equation}

If $\rho$ is a density matrix, the probability distributions related to the position and momentum basis, are
\begin{eqnarray}\label{62}
&&{\cal P}_X(r|\rho)=\bra{X;r}\rho\ket{X;r};\;\sum _{r=0}^{d-1}{\cal P}_X(r|\rho)=1\nonumber\\
&&{\cal P}_P(r|\rho)=\bra{P;r}\rho\ket{P;r};\;\sum _{r=0}^{d-1}{\cal P}_P(r|\rho)=1.
\end{eqnarray}
Below we use the following notation for the projectors
\begin{eqnarray}\label{333}
&&\Pi_X(r)=\ket{X;r}\bra{X;r};\;\;\;\Pi_P(r)=\ket{P;r}\bra{P;r}\nonumber\\
&&\sum \Pi_X(r)=\sum \Pi_P(r)={\bf 1}
\end{eqnarray}

\section{Lorenz values}

Lorenz values  are used extensively in Mathematical Economics for the study of inequality in the distribution of wealth.
We propose similar quantities for the study of inequality in the probability distributions of a quantum state, with respect to an orthonormal basis.
 Lorenz values require  ordering of the $d$ values of the probability distribution, and this leads 
to the  ordering permutation of a density matrix.

Let $\pi_X$ be the permutation for which the ${\cal P}_X(r|\rho)$ are ordered in ascending order:
\begin{eqnarray}\label{246}
{\cal P}_X(\pi_X(0)|\rho)\le ...\le {\cal P}_X(\pi_X(d-1)|\rho).
\end{eqnarray}
We refer to it as the ordering permutation of the density matrix $\rho$ with respect to the $X$-basis.
We can also define the ordering permutation $\pi_P$ with respect to the $P$-basis:
\begin{eqnarray}
{\cal P}_P(\pi_P(0)|\rho)\le  ...\le {\cal P}_P(\pi_P(d-1)|\rho).
\end{eqnarray}

The Lorenz values ${\cal L}_X(\ell;\rho)$ and ${\cal L}_P(\ell;\rho)$ are:
\begin{eqnarray}\label{LO}
&&{\cal L}_X(\ell;\rho)={\cal P}_X(\pi_X(0)|\rho)+ ...+{\cal P}_X(\pi_X(\ell)|\rho)\nonumber\\
&&{\cal L}_P(\ell;\rho)={\cal P}_P(\pi_P(0)|\rho)+ ...+{\cal P}_P(\pi_P(\ell)|\rho)\nonumber\\
\end{eqnarray}
where $\ell=0,...,d-1$. They are increasing functions of $\ell$ and 
\begin{eqnarray}\label{31}
{\cal L}_X(d-1;\rho)={\cal L}_P(d-1;\rho)=1.
\end{eqnarray}
Also
\begin{eqnarray}\label{bbb}
&&{\cal L}_X[\ell;D(\alpha, \beta)^{\dagger}\rho D(\alpha, \beta)]={\cal L}_X(\ell;\rho)\nonumber\\&&{\cal L}_P[\ell;D(\alpha, \beta)^{\dagger}\rho D(\alpha, \beta)]={\cal L}_P(\ell;\rho)\nonumber\\
&&{\cal L}_X[\ell;F^{\dagger}\rho F]={\cal L}_P(\ell;\rho)\nonumber\\&&{\cal L}_P[\ell;F^{\dagger}\rho F]={\cal L}_X(\ell;\rho)
\end{eqnarray}
We next prove that
\begin{eqnarray}\label{30}
0\le {\cal L}_X(\ell;\rho)\le \frac{\ell+1}{d}.
\end{eqnarray}
We consider two cases:
\begin{itemize}
\item[(i)]
If  ${\cal P}_X(\pi_X(\ell)|\rho)\le \frac{1}{d}$ we get
\begin{eqnarray}
{\cal L}_X(\ell;\rho)\le (\ell +1){\cal P}_X(\pi_X(\ell)|\rho) \le  \frac{\ell +1}{d}.
\end{eqnarray}
This proves Eq.(\ref{30}) for this case.
\item[(ii)]
We start from Eq.(\ref{62}) which we rewrite as
\begin{eqnarray}
{\cal P}_X(\pi_X(0)|\rho)+...+{\cal P}_X(\pi_X(d-1)|\rho)=1.
\end{eqnarray}
For $k>\ell$ we have ${\cal P}_X(\pi_X(k)|\rho)\ge{\cal P}_X(\pi_X(\ell)|\rho)$, and we replace the ${\cal P}_X(\pi_X(k)|\rho)$ with ${\cal P}_X(\pi_X(\ell)|\rho)$. We get
\begin{eqnarray}\label{86}
{\cal L}_X(\ell;\rho)+(d-\ell-1){\cal P}(X;\pi_X(\ell)|\rho)\le 1.
\end{eqnarray}

If  ${\cal P}_X(\pi_X(\ell)|\rho)>\frac{1}{d}$ we get
\begin{eqnarray}\label{8}
&&1\ge {\cal L}_X(\ell;\rho)+(d-\ell-1){\cal P}_X(\pi_X(\ell)|\rho)\nonumber\\&&\ge {\cal L}_X(\ell;\rho)+\frac{d-\ell-1}{d}.
\end{eqnarray}
From this follows Eq.(\ref{30}) for this case.
\end{itemize}

As an example we consider the case $\rho =\frac{1}{d}{\bf 1}$. 
In this case
\begin{eqnarray}\label{n1}
&&{\cal P}_X\left (r|\frac{1}{d}{\bf 1}\right)={\cal P}_P\left (r|\frac{1}{d}{\bf 1}\right )=\frac{1}{d}\nonumber\\&&{\cal L}_X\left (\ell;\frac{1}{d}{\bf 1}\right )={\cal L}_P\left (\ell;\frac{1}{d}{\bf 1}\right )=\frac{\ell +1}{d}.
\end{eqnarray}

We also prove that the $d$ density matrices corresponding to position states 
\begin{eqnarray}
\rho =\ket{X;a}\bra{X;a};\;\;\;a=0,...,d-1,
\end{eqnarray}
are the only ones for which
\begin{eqnarray}\label{e4}
&&{\cal L}_X(\ell;\rho )=0\;\;{\rm if}\;\;\ell\le d-2\nonumber\\
&&{\cal L}_X(d-1;\rho )=1.
\end{eqnarray}
Analogous result holds for the momentum states.

In order to get Eq.(\ref{e4}) we need
\begin{eqnarray}
{\cal P}_X(r|\rho)=\delta(r,a),
\end{eqnarray}
and this occurs only for position states.

\subsection{Comonotonicity}

Comonotonicity is a concept which has been used in the context of Choquet integrals with applications in Mathematical Economics and Artificial 
Intelligence\cite{D0,D1,D2,D3}, Quantum theory\cite{Q1,Q2,Q3}, etc.
Here we use comonotonicity in connection with Lorenz values in a quantum context. 

Comonotonic density matrices  have the same ordering permutation.
We show that the Lorenz values are in general superadditive functions and that they become  additive functions for comonotonic density matrices.
If $\rho$ and $\sigma$ are comonotonic density matrices and $\pi_X$ is their ordering permutation  then
$\pi_X$ is also the ordering permutation of $\lambda_1\rho+\lambda_2\sigma$ where $\lambda_1, \lambda_2$ are probabilities.

We next show that the Lorenz values are superadditive functions.
If $\lambda_1,\lambda_2$ are probabilities
\begin{eqnarray}\label{hhh}
&&{\cal L}_X(\ell;\lambda_1\rho+\lambda _2\sigma)\ge \lambda _1{\cal L}_X(\ell;\rho)+\lambda_2{\cal L}_X(\ell;\sigma)\nonumber\\
&&\lambda_1+\lambda_2=1;\;\;\;0\le \lambda_1,\lambda_2\le 1.
\end{eqnarray}

In the special case that $\rho$ and $\sigma$ are comonotonic density matrices then the inequality of Eq.(\ref{hhh}) becomes equality:
\begin{eqnarray}\label{BB1}
&&{\cal L}_X(\ell;\lambda_1\rho+\lambda _2\sigma)=\lambda_1{\cal L}_X(\ell;\rho)+\lambda_2{\cal L}_X(\ell;\sigma)\nonumber\\
&&\lambda_1+\lambda_2=1;\;\;\;0\le \lambda_1,\lambda_2\le 1.
\end{eqnarray}
Similar results hold for ${\cal L}_P(\ell;\rho)$.

In order to prove this we note that
\begin{eqnarray}\label{V1}
&&{\cal L}_X(\ell;\lambda_1\rho+\lambda _2\sigma)=
{\cal P}_X(\pi_X(0)|\lambda_1\rho+\lambda _2\sigma)+ \nonumber\\&&...+{\cal P}_X(\pi_X(\ell)|\lambda_1\rho+\lambda _2\sigma)\nonumber\\&&=
\lambda_1[{\cal P}_X(\pi_X(0)|\rho)+ ...+{\cal P}_X(\pi_X(\ell)|\rho)]\nonumber\\&&+\lambda_2[{\cal P}_X(\pi_X(0)|\sigma)+ ...+{\cal P}_X(\pi_X(\ell)|\sigma)]
\end{eqnarray}
$\widetilde \pi_X$ is the ordering permutation for $\lambda_1\rho+\lambda _2\sigma$, and in general it will not be the 
ordering permutation for $\rho$ and $\sigma$. Therefore
\begin{eqnarray}\label{V2}
&&{\cal P}_X(\widetilde \pi_X(0)|\rho)+ ...+{\cal P}_X(\widetilde \pi_X(\ell)|\rho)\ge {\cal L}_X(\ell;\rho)\nonumber\\
&&{\cal P}_X(\widetilde \pi_X(0)|\sigma)+ ...+{\cal P}_X(\widetilde \pi_X(\ell)|\sigma)\ge {\cal L}_X(\ell;\sigma)\nonumber\\
\end{eqnarray}
Combining Eqs(\ref{V1}), (\ref{V2}) we prove Eq.(\ref{hhh}).

In the special case of comonotonic $\rho, \sigma$, the $\rho$, $\sigma$,  $\lambda_1\rho+\lambda _2\sigma$ have the same ordering permutation and the inequalities in Eq.(\ref{V2}) become equalities.
From this follows Eq.(\ref{BB1}).

\section{The Gini index  and the uncertainty principle}

The Gini index with respect to the position basis is defined
\begin{eqnarray}\label{84}
&&{\cal G}_X(\rho)=\frac{1}{\cal N}\sum _{\ell=0}^{d-1}\left[{\cal L}_X\left (\ell;\frac{1}{d}{\bf 1}\right )-{\cal L}_X(\ell;\rho)\right ]\nonumber\\&&=\frac{1}{\cal N}\sum _{\ell=0}^{d-1}\left[\frac{\ell +1}{d}-{\cal L}_X(\ell;\rho)\right ]\nonumber\\
&&{\cal N}=\sum _{\ell =0}^{d-1}{\cal L}_X \left (\ell;\frac{1}{d}{\bf 1}\right )=\sum _{\ell=0}^{d-1}\frac{\ell+1}{d}=\frac{d+1}{2}\nonumber\\
\end{eqnarray}
Another equivalent definition is
\begin{eqnarray}\label{85}
{\cal G}_X(\rho)&=&1-\frac{2}{d+1}\sum _{\ell=0}^{d-1}{\cal L}_X(\ell;\rho)\nonumber\\&=&\frac{d-1}{d+1}-\frac{2}{d+1}\sum _{\ell=0}^{d-2}{\cal L}_X(\ell;\rho)\nonumber\\&=&
1-\frac{2}{d+1}[d{\cal P}_X(\pi_X(0)|\rho)\nonumber\\&+&(d-1){\cal P}_X(\pi_X(1)|\rho)+...\nonumber\\&+&{\cal P}_X(\pi_X(d-1)|\rho)].
\end{eqnarray}
In a similar way to ${\cal G}_X(\rho)$ we define the Gini index with respect to the momentum basis ${\cal G}_P(\rho)$.
We also define the sum
\begin{eqnarray}
{\cal G}_{XP}(\rho)={\cal G}_X(\rho)+{\cal G}_P(\rho).
\end{eqnarray}

The Gini index ${\cal G}_X(\rho)$ quantifies  how close is the ${\cal P}(X;r|\rho)$ to a uniform distribution (which describes maximum uncertainty).
Small (large) values ${\cal G}_X(\rho)$ indicate large (small) uncertainty in the probability distribution ${\cal P}(X;r|\rho)$.

\begin{proposition}\label{pro56}
\mbox{}
\begin{itemize}
\item[(1)]
\begin{eqnarray}
0\le {\cal G}_X(\rho)\le\frac{d-1}{d+1};\;\;\;0\le {\cal G}_P(\rho)\le\frac{d-1}{d+1}.
\end{eqnarray}
The ${\cal G}_X(\rho)$ indicates the uncertainty in the outcome with the measurements $\Pi_X(r)$.
${\cal G}_X(\rho)=\frac{d-1}{d+1}$ indicates a certain outcome, while 
${\cal G}_X(\rho)=0$ indicates the most uncertain outcome.
Similar comment can be made for ${\cal G}_P(\rho)$ for the measurements $\Pi_P(r)$.
\item[(2)]
${\cal G}_X(\rho)=\frac{d-1}{d+1}$ 
only for the $d$ density matrices $\rho =\ket{X;a}\bra{X;a}$.
Similarly, ${\cal G}_P(\rho)=\frac{d-1}{d+1}$ 
only for the $d$ density matrices $\rho =\ket{P;a}\bra{P;a}$.

\item[(3)]
It is impossible to have ${\cal G}_X(\rho)={\cal G}_P(\rho)=\frac{d-1}{d+1}$. Therefore
\begin{eqnarray}\label{PP}
0\le {\cal G}_{XP}(\rho)<2\frac{d-1}{d+1}.
\end{eqnarray}
In fact ${\cal G}_{XP}(\rho)$ cannot take values which are arbitrarily close to $2\frac{d-1}{d+1}$.
\end{itemize}
\end{proposition}
\begin{proof}
\mbox{}
\begin{itemize}
\item[(1)]
From Eqs.(\ref{30}),(\ref{84}) follows that $0\le {\cal G}_X(\rho)$.
Then {Eq.(\ref{85})} and the fact that the ${\cal L}_X(\ell;\rho)$ are non-negative numbers, prove that ${\cal G}_X(\rho)\le\frac{d-1}{d+1}$.

\item[(2)]
{Eq.(\ref{85})} shows that we get ${\cal G}_X(\rho)=\frac{d-1}{d+1}$  only if
\begin{eqnarray}
&&{\cal L}_X(\ell;\rho )=0\;\;{\rm if}\;\;\ell\le d-2\nonumber\\
&&{\cal L}_X(d-1;\rho )=1.
\end{eqnarray}
We have seen earlier that this occurs 
only for the $d$ density matrices $\rho =\ket{X;a}\bra{X;a}$.
\item[(3)]
We have just proved that ${\cal G}_X(\rho)=\frac{d-1}{d+1}$  only 
for the $d$ density matrices $\rho =\ket{X;a}\bra{X;a}$, in which case ${\cal G}_P(\rho)=0$.
Similarly ${\cal G}_P(\rho)=\frac{d-1}{d+1}$  only 
for the $d$ density matrices $\rho =\ket{P;a}\bra{P;a}$, in which case ${\cal G}_X(\rho)=0$. Therefore it is impossible to have ${\cal G}_X(\rho)={\cal G}_P(\rho)=\frac{d-1}{d+1}$
and from this follows Eq.(\ref{PP}).

We next assume that a density matrix $\rho$ has 
\begin{eqnarray}\label{890}
{\cal G}_X(\rho)=\frac{d-1}{d+1}-\epsilon _X;\;\;\;{\cal G}_P(\rho)=\frac{d-1}{d+1}-\epsilon _P
\end{eqnarray}
where $\epsilon _X$ and $\epsilon _P$ are non-negative infinitesimals, so that ${\cal G}_{XP}(\rho)$ is arbitrarily close to $2\frac{d-1}{d+1}$.
In this case ${\cal L}(\ell,\rho)=\epsilon _{\ell X}$ where $\epsilon _{\ell X}$ are non-negative infinitesimals and
\begin{eqnarray}\label{850}
\frac{2}{d+1}\sum _{\ell=0}^{d-2}\epsilon _{\ell X}=\epsilon _X.
\end{eqnarray}
In this case one of the ${\cal P}_X(r|\rho)$ is equal to $1-\epsilon_X$, and the others are non-negative infinitesimals $e _{Xr}$ (with $\sum e_{Xr}=\epsilon_X$).
It follows that $\rho$ can be written as
\begin{eqnarray}
\rho=(1-\epsilon _X)\ket{X;a}\bra{X;a}+\tau;\;\;\;{\rm Tr}(\tau)=\epsilon_X.
\end{eqnarray}
where $\tau$ is not a density matrix in general. Indeed in this case
\begin{eqnarray}
&&{\cal P}_X(r|\rho)=(1-\epsilon_X)\delta(r,a)+e_{Xr}\nonumber\\&&e_{Xr}=\bra{X;r}\tau\ket{X;r}.
\end{eqnarray}
But then
\begin{eqnarray}
&&{\cal P}_P(r|\rho)=(1-\epsilon_X)\frac{1}{d}+{\widetilde e}_{Xr}\nonumber\\&&{\widetilde e}_{Xr}=\bra{P;r}\tau\ket{P;r}.
\end{eqnarray}
The ${\widetilde e}_{Xr}$ are non-negative infinitesimals, because the $e _{Xr}$ are non-negative infinitesimals.
Therefore 
\begin{eqnarray}
{\cal L}_P(\ell,\rho)=(1-\epsilon_X)\frac{\ell+1}{d}+\epsilon_{\ell P}
\end{eqnarray}
where $\epsilon _{\ell P}$ are non-negative infinitesimals.
But in this case ${\cal G}_P(\rho)$ will have an infinitesimal value rather than the value given in Eq.(\ref{890}).
This proves that ${\cal G}_{XP}(\rho)$ cannot take values which are arbitrarily close to $2\frac{d-1}{d+1}$.
\end{itemize}
\end{proof}
The above proposition shows that it is impossible to have certain outcome with the measurements $\Pi_X(r)$ on an ensemble described by the density matrix $\rho$, and also certain outcome
with the measurements $\Pi_P(r)$ on another ensemble described by the same density matrix $\rho$.
This is an uncertainty principle in terms of the Gini index, which is quantified below.

We next show that
\begin{eqnarray}\label{300}
&&{\cal G}_{XP}[D(\alpha, \beta)^{\dagger}\rho D(\alpha, \beta)]={\cal G}_{XP}(\rho)\nonumber\\
&&{\cal G}_{XP}[F^{\dagger}\rho F]={\cal G}_{XP}(\rho)
\end{eqnarray}
Using Eq.(\ref{bbb})  we prove that
\begin{eqnarray}
&&{\cal G}_X[D(\alpha, \beta)^{\dagger}\rho D(\alpha, \beta)]={\cal G}_X(\rho)\nonumber\\&&{\cal G}_P[D(\alpha, \beta)^{\dagger}\rho D(\alpha, \beta)]={\cal G}_P(\rho)\nonumber\\
&&{\cal G}_X[F^{\dagger}\rho F]={\cal G}_P(\rho)\nonumber\\&&{\cal G}_P[F^{\dagger}\rho F]={\cal G}_X(\rho).
\end{eqnarray}
From this follows Eq.(\ref{300}).

Using this Eq.(\ref{300}) with $\rho=\ket{f}\bra{f}$ we find that
all $d^2$ coherent states of Eq.(\ref{bba}) have the same ${\cal G}_{XP}(\rho)$ (which depends on the fiducial vector):
\begin{eqnarray}
{\cal G}_{XP}(\ket{\alpha, \beta}_{\rm coh}  \;_{\rm coh}\bra{\alpha, \beta})={\cal G}_{XP}(\ket{f}\bra{f}).
\end{eqnarray}

As an example we consider the density matrix $\rho=\frac{1}{d}{\bf 1}$ we get ${\cal G}_X\left (\frac{1}{d}{\bf 1}\right)={\cal G}_P\left (\frac{1}{d}{\bf 1}\right)=0$ and therefore
\begin{eqnarray}
{\cal G}_{XP}\left (\frac{1}{d}{\bf 1}\right)=0.
\end{eqnarray}
In this sense this is the density matrix with maximum uncertainty.

We also consider the density matrix $\rho =\ket{X;a}\bra{X;a}$ and we get ${\cal G}_X(\ket{X;a}\bra{X;a})=\frac{d-1}{d+1}$ and ${\cal G}_P\left (\ket{X;a}\bra{X;a}\right )=0$. Therefore
\begin{eqnarray}
{\cal G}_{XP}(\ket{X;a}\bra{X;a})=\frac{d-1}{d+1}.
\end{eqnarray}
Similarly
\begin{eqnarray}
{\cal G}_{XP}(\ket{P;a}\bra{P;a})=\frac{d-1}{d+1}.
\end{eqnarray}

We next show that the Gini index is a subadditive function.
If $\lambda_1, \lambda_2$ are probabilities, the following inequality holds:
\begin{eqnarray}\label{D1}
&&{\cal G}_X(\lambda_1\rho+\lambda _2\sigma)\le \lambda_1{\cal G}_X(\rho)+\lambda_2{\cal G}_X(\sigma)\nonumber\\
&&\lambda_1+\lambda_2=1;\;\;\;0\le \lambda_1,\lambda_2\le 1.
\end{eqnarray}
In the special case of $\rho$ and $\sigma$ are comonotonic density matrices with respect to the $X$-basis,  this becomes equality:
\begin{eqnarray}\label{BB2}
{\cal G}_X(\lambda_1\rho+\lambda _2\sigma)=\lambda_1{\cal G}_X(\rho)+\lambda_2{\cal G}_X(\sigma)
\end{eqnarray}
Analogous result holds for ${\cal G}_P(\rho)$.

To prove this we note that {Eq.(\ref{D1})} follows from Eq.(\ref{hhh}). For comonotonic density matrices we use Eq.(\ref{BB1}) to prove Eq.(\ref{BB2}).

From Eq.(\ref{D1}) follows immediately that ${\cal G}_{XP}(\rho)$ is a subadditive function.
If $\lambda_1, \lambda_2$ are probabilities, the following inequality holds:
\begin{eqnarray}\label{D10}
&&{\cal G}_{XP}(\lambda_1\rho+\lambda _2\sigma)\le \lambda_1{\cal G}_{XP}(\rho)+\lambda_2{\cal G}_{XP}(\sigma)\nonumber\\
&&\lambda_1+\lambda_2=1;\;\;\;0\le \lambda_1,\lambda_2\le 1.
\end{eqnarray}

Therefore
\begin{eqnarray}\label{sdf}
&&{\cal G}_{XP}\left (\sum \lambda_a\ket{X;a}\bra{X;a}\right)\le \frac{d-1}{d+1}\nonumber\\&&\sum \lambda_a=1;\;\;\lambda_a\ge 0.
\end{eqnarray}

The uncertainty coefficient is
\begin{eqnarray}\label{UN}
\eta (d)=2\frac{d-1}{d+1}-{\mathfrak G}(d),
\end{eqnarray}
where  ${\mathfrak G}(d)$ is the supremum of ${\cal G}_{XP}(\rho)$ over the { set ${\mathfrak R}$ of all density matrices}:
\begin{eqnarray}
{\mathfrak G}(d)=\sup _{\rho \in {\mathfrak R}}{\cal G}_{XP}(\rho).
\end{eqnarray}

\begin{proposition}\label{pro134}
\begin{eqnarray}\label{F2}
0<\eta (d)\le \frac{d-1}{d+1}\frac{\sqrt{d}}{1+\sqrt{d}}.
\end{eqnarray}
\end{proposition}
\begin{proof}
Taking into account the value of ${\cal G}_{XP}(\rho)$ for an example in Eq.(\ref{100}) below, and also proposition \ref{pro56}, we conclude that
\begin{eqnarray}\label{F1}
\frac{d-1}{d+1}\left (1+\frac{1}{1+\sqrt{d}}\right)\le{\mathfrak G}(d)< 2\frac{d-1}{d+1}.
\end{eqnarray}
It is crucial that ${\cal G}_{XP}(\rho)$ cannot take values which are arbitrarily close to $2\frac{d-1}{d+1}$ (proposition \ref{pro56}), and therefore
${\mathfrak G}(d)$ cannot be equal to $2\frac{d-1}{d+1}$.
From Eq.(\ref{F1}) follows Eq.(\ref{F2}).
\end{proof}

Proposition \ref{pro134} quantifies the uncertainty principle in the context of quantum systems with finite-dimensional Hilbert space.
$\eta (d)=0$ would mean that there exists some density matrix that gives
outcome with no uncertainty with the measurements $\Pi_X(r)$, and also outcome with no uncertainty 
with the measurements $\Pi_P(r)$.
We have shown that $\eta (d)$ is a positive number,
and therefore there is non-zero uncertainty in the outcome of at least one of these two measurements.

Eq.~(\ref{F2}) gives an upper bound for $\eta(d)$.
It is not easy to find the precise value of $\eta (d)$, and this requires further work.
\section{Example}

We consider the density matrix 
\begin{eqnarray}
\rho=\ket{s}\bra{s};\;\;\;\ket{s}=\frac{d^{1/4}}{\sqrt{2d^{1/2}+2}}[\ket{X;0}+\ket{P;0}]
\end{eqnarray}
In this case
{
\begin{eqnarray}
&&{\cal P}_X(\pi_X(r)|\rho)=\frac{1}{2d+2\sqrt{d}};\;\;\;r=0,...,d-2\nonumber\\
&&{\cal P}_X(\pi_X(d-1)|\rho)=\frac{d+1+2\sqrt{d}}{2d+2\sqrt{d}}
\end{eqnarray}}
Then
\begin{eqnarray}
&&{\cal L}_X(\ell;\rho )=\frac{\ell+1}{2d+2\sqrt{d}}\;\;{\rm if}\;\;\ell \le d-2\nonumber\\
&&{\cal L}_X(d-1;\rho )=1
\end{eqnarray}
Therefore
\begin{eqnarray}
{\cal G}_X(\rho )=\frac{d-1}{d+1}\frac{2+\sqrt{d}}{2+2\sqrt{d}}.
\end{eqnarray}
Due to symmetry ${\cal G}_P(\rho )={\cal G}_X(\rho )$ and therefore
\begin{eqnarray}\label{100}
{\cal G}_{XP}(\rho )=\frac{d-1}{d+1}\left (1+\frac{1}{1+\sqrt{d}}\right).
\end{eqnarray}
This has already been used in Eq.(\ref{F1}).

Taking into account Eq.(\ref{300}) we conclude that the density matrices $\rho (\alpha, \beta)=\ket{s}\bra{s}$ where
\begin{eqnarray}
\ket{s}=\frac{1}{\sqrt{2+\frac{2}{\sqrt{d}}}}D(\alpha, \beta)[\ket{X;0}+\ket{P;0}]
\end{eqnarray}
also have the ${\cal G}_{XP}(\rho )$ given in Eq.(\ref{100}).

\section{Discussion}
We studied uncertainty relations for systems with finite-dimensional Hilbert space, using 
the Lorenz values and the Gini index.

Lorenz values require an ordering of the $d$ values of the probability distribution of a density matrix with respect to an orthonormal basis. This leads 
to the  ranking permutation of a density matrix, and to comonotonic density matrices (which have the same ranking permutation).
Lorenz values are defined in Eq.(\ref{LO}).
It is shown that the Lorenz values are superadditive functions in general, and that for comonotonic density matrices they become additive functions.

The Gini index is defined in {Eqs.(\ref{84}),(\ref{85})} and its properties are discussed.
It is shown that the Gini index is an subadditive functions in general, and that for comonotonic density matrices it becomes an additive function.

The uncertainty coefficient is defined in Eq.(\ref{UN}). It is proven that it is positive number, and an upper bound has been given in Eq.~(\ref{F2}).
The uncertainty coefficient quantifies the uncertainty principle for quantum systems with finite-dimensional Hilbert space.

The work brings the Lorenz values and the Gini index, in the context of Quantum Physics.

\end{document}